\documentclass[conference,a4paper]{IEEEtran}
\usepackage{amsmath,amsthm,amssymb,color,graphicx,caption,subcaption,comment,tikz}

\title{Capacity Bounds for a Class of Diamond Networks}
\author{
  \IEEEauthorblockN{Shirin Saeedi Bidokhti}
  \IEEEauthorblockA{
  Lehrstuhl f\"{u}r Nachrichtentechnik\\
    Technische Universit\"{a}t M\"{u}nchen, Germany\\
    shirin.saeedi@tum.de} 
  \and
  \IEEEauthorblockN{Gerhard Kramer}
  \IEEEauthorblockA{
  Lehrstuhl f\"{u}r Nachrichtentechnik\\
    Technische Universit\"{a}t M\"{u}nchen, Germany\\
    gerhard.kramer@tum.de}
}
\date{}

\newtheorem{theorem}{Theorem}

\newtheorem{lemma}{Lemma}
\newtheorem{proposition}{Proposition}
\newtheorem{corollary}{Corollary}

\newtheorem{remark}{Remark}

\begin{document}

\maketitle
\begin{abstract}
A class of diamond networks are studied where  the broadcast component is modelled by two independent bit-pipes. New upper and low bounds are derived on the capacity which improve previous bounds. The upper bound is in the form of a max-min problem, where the maximization is over a coding distribution and the minimization is over an auxiliary channel. The proof technique generalizes bounding techniques of Ozarow for the  Gaussian multiple description problem (1981), and Kang and Liu for the Gaussian diamond network  (2011). The bounds are evaluated for a Gaussian multiple access channel (MAC) and the binary adder MAC, and the capacity is found for interesting ranges of the bit-pipe capacities.
\end{abstract}
\section{Introduction}

The diamond network was introduced in \cite{Schein01} as a simple multi-hop network with broadcast and multiple access channel (MAC) components. The  two-relay diamond network models the scenario where a source communicates with a sink through two relay nodes that do not have information of their own to communicate. The underlying challenge may be described as follows. In order to fully utilize the MAC to the receiver, we would ideally like full cooperation between the relay nodes. On the other hand, in order to communicate the maximum amount of information and better use the diversity that is offered by the relays, we would like to send independent information to the relay nodes over the broadcast channel. The  problem of finding the capacity of this network is unresolved. Lower and upper bounds  on the capacity are given in \cite{Schein01}. We remark that the problem is solved over linear deterministic relay networks, and the capacity of Gaussian relay networks has been approximated within a constant number of bits \cite{AvestimehrDiggaviTse11}.

In this paper, we study a class of diamond networks where the broadcast channel is modelled by two independent bit-pipes.
This problem was initially studied in \cite{TraskovKramer07} where lower and upper bounds were derived on the capacity. The best known upper bound for the problem is the cut-set bound, which does not fully capture the tradeoff between cooperation and diversity. Recently, \cite{KangLiu11} studied the network with a Gaussian MAC and proved a new upper bound which constrains the mutual information between the MAC inputs to improve the cut-set bound. The bounding technique in \cite{KangLiu11} is motivated by \cite{Ozarow80} that treats the Gaussian multiple description problem. Unfortunately, neither result seems to apply to discrete memoryless channels. 

This paper is organized as follows. We state the problem setup in Section \ref{prel}. In Section \ref{secupp}, we prove a new upper bound on the achievable rates by generalizing the bounding technique of \cite{KangLiu11}. Our upper bound applies to the general class of discrete memoryless MACs, and strictly improves the cut-set bound. In Section \ref{secach}, we improve the achievable rates of \cite{TraskovKramer07} by communicating a common piece of information from the source to both relays using superposition coding and Marton's coding. Finally, we study our bounds for networks with a Gaussian MAC (Section \ref{exGauss}) and a binary adder MAC (Section \ref{exAdd}). For both examples, we find conditions on the bit-pipe capacities such that the upper and lower bounds meet.
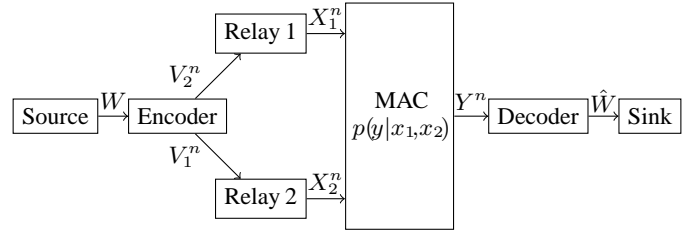
\begin{figure}[t!]
\centering
\begin{tikzpicture}[scale=.73]
\tikzstyle{every node}=[draw,shape=circle,font=\small];

\path (-1.2,1.5) node[shape=rectangle] (s0) {$\text{Source}$};
\path (1,1.5) node[shape=rectangle] (s) {$\text{Encoder}$};

\path (2.5,3) node[shape=rectangle] (r2) {$\text{Relay\hspace{-.05cm} 1}$};
\path (2.5,0) node[shape=rectangle] (r1) {$\text{Relay\hspace{-.05cm} 2}$};
\path (5,1.5) node[shape=rectangle,minimum height=3cm] (mac) {$\hspace{-.2cm}\begin{array}{c}\text{MAC}\\p(\hspace{-.05cm}y|x_1\hspace{-.05cm},\hspace{-.05cm}x_2\hspace{-.05cm})\end{array}\hspace{-.25cm}$};

\path (7.5,1.5) node[shape=rectangle] (d) {$\text{Decoder}$};
\path (9.51,1.5) node[shape=rectangle] (ds) {$\text{Sink}$};

\path (4.25,0) node[draw=none] (inp1) {};
\path (4.25,3) node[draw=none] (inp2) {};
\path (5.75,1.5) node[draw=none] (out) {};

\draw[->] (s0) --node[draw=none,yshift=.2cm]{$W$} (s);  
\draw[->] (d) --node[draw=none,yshift=.2cm]{$\hat{W}$} (ds);  
\draw[->] (s) --node[left,draw=none]{$V^n_1$} (r1);
\draw[->] (s) --node[left,draw=none]{$V^n_2$} (r2);
\draw[->] (r2) --node[draw=none,yshift=.2cm]{$X^n_1$} (inp2);
\draw[->] (r1) --node[draw=none,yshift=.2cm]{$X^n_2$} (inp1);
\draw[->] (out) --node[draw=none,yshift=.2cm]{$Y^n$} (d);

\end{tikzpicture}
\caption{Problem setup}
\label{two-user-mac}
\end{figure}

\section{Preliminaries}
\label{prel}
\subsection{Notation}
We use standard notation for random variables, e.g., $X$, probabilities, e.g., $p_X(x)$ or $p(x)$, entropies, e.g., $H(X)$ and $H(X|Y)$, and mutual information, e.g., $I(X;Y)$. We denote the sequence $X_1,\ldots,X_n$ by $X^n$. Sets are denoted by script letters and matrices are denoted by bold capital letters.
\subsection{Model}
Consider the diamond network in Fig. \ref{two-user-mac}.  A source communicates a message of rate $R$ to a sink.  The source is connected  to two relays via noiseless bit-pipes of capacities $C_1$ and $C_2$,  and the relays communicate with the receiver over a MAC.

The source encodes a message $W$ with $nR$ bits into sequence $V_1^n$, which is available at encoder $1$, and sequence $V_2^n$, which is available at encoder $2$. $V_1^n$ and $V_2^n$ are such that $H(V_1^n)\leq nC_1$ and $H(V_2^n)\leq nC_2$.
Each relay $i$, $i=1,2$,  maps its received sequence $V_i^n$ into a sequence $X_i^n$ which is sent over the MAC.
The MAC is characterized by its input alphabets $\mathcal{X}_1$ and $\mathcal{X}_2$,  output alphabet $\mathcal{Y}$, and transition probabilities $p(y|x_1,x_2)$, for each $x_1\in\mathcal{X}_1$, $x_2\in\mathcal{X}_2$.
From the received sequence, the sink decodes an estimate $\hat{W}$ of $W$.
We are interested in characterizing the highest rate $R$ that permits arbitrarily small positive error probability $\Pr(\hat{W}\neq W)$. 


\section{An upperbound}
\label{secupp}
The idea behind our upper bound is motivated by  \cite{KangLiu11,Ozarow80}. The proposed bound is applicable not only to Gaussian channels, but also to general discrete memoryless channels, and it strictly improves the cut-set bound as we show via two examples.

In the network of Fig. \ref{two-user-mac}, the cut-set bound gives
\begin{eqnarray}
R&\leq&C_1+C_2\label{disregard}\\
R&\leq&C_1+I(X_2;Y|X_1)\\
R&\leq&C_2+I(X_1;Y|X_2)\\
R&\leq&I(X_1X_2;Y).
\end{eqnarray}

However, the bound in \eqref{disregard} disregards the potential correlation between the inputs of the MAC.
A tighter bound on $R$ is given in the following multi-letter form:
\begin{eqnarray*}
nR&\leq& H(V_1^n,V_2^n)\\&=& H(V_1^n)+H(V_1^n)-I(V_1^n;V_2^n)\\&\leq&nC_1+nC_2-I(X_1^n;X_2^n).
\end{eqnarray*}
It is observed in \cite{KangLiu11} that $I(X_1^n;X_2^n)$ can be written in the following form for  \textit{any} random sequence $U^n$.
\begin{eqnarray*}
I(X_1^n;X_2^n)&=& I(X_1^nX_2^n;U^n)-I(X_1^n;U^n|X_2^n)\\&&-I(X_2^n;U^n|X_1^n)+I(X_1^n;X_2^n|U^n)
\end{eqnarray*}
Therefore, we have
\begin{eqnarray}
nR&\leq&nC_1+nC_2-I(X_1^nX_2^n;U^n)\nonumber\\&&+I(X_1^n;U^n|X_2^n)+I(X_2^n;U^n|X_1^n).\label{eqref}
\end{eqnarray}

It may not be clear how to single-letterize the above bound. We proceed as follows.
First, note that
\begin{eqnarray}
nR&\leq&I(X_1^nX_2^n;Y^n)\leq I(X_1^nX_2^n;Y^nU^n).\label{eq4}
\end{eqnarray}
Combining inequalities  \eqref{eqref} and \eqref{eq4}, we have
\begin{eqnarray}
2nR&\leq&nC_1+nC_2+I(X_1^nX_2^n;Y^n|U^n)\nonumber\\&&+I(X_1^n;U^n|X_2^n)+I(X_2^n;U^n|X_1^n).
\end{eqnarray}

Define $U_i$ from $X_{1i}, X_{2i}, Y_i$ through the channels $p_{U_i|X_{1i} X_{2i} Y_i}(u_i|x_{1i},x_{2i},y_i)$, $i=1,2,\ldots,n$. With this choice of $U_i$ we have the following chain of inequalities:
\allowdisplaybreaks
\begin{eqnarray*}
&&2nR\\&&\leq nC_1+nC_2+I(X_1^nX_2^n;Y^n|U^n)\\&&\quad+I(X_1^n;U^n|X_2^n)+I(X_2^n;U^n|X_1^n)\\
&&\leq nC_1+nC_2+\sum_{i=1}^nI(X_1^nX_2^n;Y_i|U^nY^{i-1})\\&&\quad+\sum_{i=1}^nI(X_1^n;U_i|X_2^nU^{i-1}\!)\!+\!\sum_{i=1}^nI(X_2^n;U_i|X_1^nU^{i-1}\!)\\
&&\stackrel{(a)}{\leq} nC_1+nC_2+\sum_{i=1}^nI(X_{1i}X_{2i};Y_i|U_i)\\&&\quad+\sum_{i=1}^nI(X_{1i};U_i|X_{2i})+\sum_{i=1}^nI(X_{2i};U_i|X_{1i})\\
&&\leq nC_1+nC_2+nI(X_{1,Q},X_{2,Q};Y_Q|U_Q)\\&&\quad+nI(X_{1,Q};U_Q|X_{2,Q})+nI(X_{2,Q};U_Q|X_{1,Q}).
\end{eqnarray*}
Step $(a)$ follows because we have the following two Markov chains:
\begin{eqnarray}
&&(X_1^nX_2^nU^nY^{i-1})-(X_{1i}X_{2i}U_i)-Y_i\\
&&(X_1^nX_2^nU^{i-1})-(X_{1i}X_{2i})-U_i.
\end{eqnarray}

We thus have the following upper bound.
\begin{theorem}
\label{thmupp}
The rate $R$ is achievable only if  there exists a joint distribution $p(x_1,x_2)$ for which inequalities \eqref{thmupp1}-\eqref{thmupp5} hold for every auxiliary channel $p(u|x_1,x_2,y)$.
\begin{eqnarray}
R&\leq& C_1+C_2\label{thmupp1}\\
R&\leq& C_2+I(X_1;Y|X_2)\label{thmupp2}\\
R&\leq& C_1+I(X_2;Y|X_1)\label{thmupp3}\\
R&\leq&I(X_1X_2;Y)\label{thmupp4}\\
2R&\leq&C_1+C_2+I(X_{1}X_{2};Y|U)\nonumber\\&&+I(X_{1};U|X_{2})+I(X_{2};U|X_{1})\label{thmupp5}
\end{eqnarray}
\end{theorem}
\begin{remark}
The bound of Theorem \ref{thmupp} is in the form of a max-min problem, where the maximization is over $p(x_1,x_2)$ and the minimization is over $p(u|x_1,x_2,y)$. \end{remark}
\begin{remark}
\label{concave}
For a fixed auxiliary channel $p(u|x_1,x_2,y)$ and a fixed MAC $p(y|x_1,x_2)$, the upper bound in Theorem \ref{thmupp} is concave in $p(x_1,x_2)$. See \cite{reportMAC}.
\end{remark}
\begin{remark}
The right hand side (RHS) of \eqref{thmupp5} may be written as
\begin{eqnarray}
\label{equivthm5}
C_1+C_2+I(X_{1}X_{2};YU)-I(X_{1};X_2)+I(X_{1};X_2|U).
\end{eqnarray}
\end{remark}
We study the upper bound of Theorem \ref{thmupp} with a Gaussian MAC and a binary adder MAC in Sections \ref{exGauss} and \ref{exAdd} respectively. 

\section{A lower-bound}
\label{secach}
The achievable scheme that we present here is based on \cite{TraskovKramer07}, but we further allow a common message at both relaying nodes. This is  done by rate splitting, superposition coding and Marton's coding and is summarized in the following theorem\footnote{ The method to derive this theorem was discussed in general terms by G. Kramer and S. Shamai after G. Kramer presented the paper \cite{TraskovKramer07} at ITW in 2007.}.
\begin{theorem}
\label{lower bound}
The rate $R$ is achievable if it satisfies the following conditions for some pmf $p(u,x_1,x_2,y)$, where $p(u,x_1,x_2,y)=p(u,x_1,x_2) p(y|x_1,x_2)$, and $U\in\mathcal{U}$ with $|\mathcal{U}|\leq \min\{|\mathcal{X}_1||\mathcal{X}_2|+2,|\mathcal{Y}|+4\}$.
\begin{eqnarray}
\label{inner}
R\!\leq\! \min\! \left\{\!\!\!\!\begin{array}{l}C_1+C_2-I(X_1;X_2|U)\\C_2+I(X_1;Y|X_2U)\\C_1+I(X_2;Y|X_1U)\\\frac{1}{2}(C_1\!+\!C_2\!+\!I(X_1X_2;Y|U)\!-\!I(X_1;X_2|U)\hspace{-.05cm})\\I(X_1X_2;Y)\end{array}\!\!\!\!\right\}
\end{eqnarray}
\end{theorem}

\begin{proof}[Sketch of proof]
The inner bound in Theorem \ref{lower bound} is found using the following achievable scheme.
\paragraph{Codebook construction}
Fix the joint pmf $P_{UX_1X_2}(u,x_1,x_2)$. Let $R=R_{12}+R^\prime$ and let   \begin{eqnarray}\label{ratesplit}R^\prime=R_1+R_2-I(X_1;X_2|U)-\delta_\epsilon.\end{eqnarray}
Generate $2^{nR_{12}}$ sequences $u^n(m_{12})$  independently, each in an i.i.d manner according to $\prod_lP_{U}(u_l)$.
For each sequence $u^n(m_{12})$, generate (i) $2^{nR_1}$ sequences $x^n_1(m_{12},m_1)$ (conditionally) independently, each in an i.i.d manner according to $\prod_lP_{X_1|U}(x_{1,l}|u_l)$ and (ii) $2^{nR_2}$ sequences $x^n_2(m_{12},m_2)$ (conditionally) independently, each in an i.i.d manner according to $\prod_lP_{X_2|U}(x_{2,l}|u_l)$.
For each sequence $u^n(m_{12})$, pick $2^{nR^\prime}$ sequence pairs $(x^n_1(m_{12},m_1),x^n_2(m_{12},m_2))$ that are jointly typical. Index such pairs by $m^\prime\in\{1,\ldots,2^{nR^\prime}\}$.

\paragraph{Encoding}
To communicate message $W=(m_{12},m^\prime)$, communicate $(m_{12},m_1)$ with relay $1$  and $(m_{12},m_2)$ with relay $2$ (we assume $(x^n_1(m_{12},m_1),x^n_2(m_{12},m_2))$ is the ${m^\prime}\text{th}$ jointly typical pair for $m_{12}$). Relays $1$ and $2$ then send  $x^n_1(m_{12},m_1)$ and $x^n_2(m_{12},m_2)$ over the MAC, respectively.
\paragraph{Decoding}
Upon receiving $y^n$, the receiver looks for $(\hat{m}_{12},\hat{m}^\prime)$ for which the following tuple is jointly typical: $$(u^n(\hat{m}_{12}),x^n_1(\hat{m}_{12},\hat{m}_1),x^n_2(\hat{m}_{12},\hat{m}_2),y^n)\in\mathcal{T}^n_\epsilon.$$
\paragraph{Error Analysis}
The error analysis is straightforward and we omit it for the sake of brevity. The above scheme performs reliably if
\begin{eqnarray}
&& R_{12}+R_1< C_1\label{C1}\\
&& R_{12}+R_2< C_2\label{C2}\\
&& R< I(UX_1X_2;Y)\label{R}\\
&& R_1+R_2-I(X_1;X_2|U)<I(X_1X_2;Y|U)\label{R1+R2}\\
&&R_2< I(X_2;Y|X_1,U)+I(X_1;X_2|U)\label{R2}\\
&&R_1< I(X_1;Y|X_2,U)+I(X_1;X_2|U).\label{R1}
\end{eqnarray}
Conditions \eqref{ratesplit}-\eqref{R1}, together with $R_1\geq 0$, $R_2\geq 0$, $R_{12}\geq 0$, characterize an achievable rate. Eliminating $R_{12},R_1,R_2$, we arrive at Theorem \ref{lower bound}. Cardinality bounds are derived using the standard method via Caratheodory's theorem \cite{ElGamalKim}.
\end{proof}
\begin{proposition}
\label{thmconcave}
The lower bound of Theorem \ref{lower bound} is concave in $C_1,C_2$.
\end{proposition}
\begin{proof}
We prove the statement for the case of $C_1=C_2=C$. The same argument holds for the more general case.
We express the lower bound of Theorem \ref{lower bound} in terms of the maximization problem $\max_{p(u,x_1,x_2)}f^p_\ell(C,p(u,x_1,x_2))$ and we denote it by $f_\ell(C)$. In this formulation, $f^p_\ell(C,p(u,x_1,x_2))$ is just the minimum term on the right hand of \eqref{inner}.

The proof of the theorem is by contradiction. Suppose that the lower bound is not concave in $C$; i.e., there exist  values $C^{(1)},C^{(2)}$, and $\alpha$, $0\leq\alpha\leq 1$, such that $C^\star=\alpha C^{(1)}+(1-\alpha)C^{(2)}$ and $f_\ell(C^\star)<\alpha f_\ell(C^{(1)})+(1-\alpha)f_\ell(C^{(2)})$. 
Let $p^{(1)}(u,x_1,x_2)$ (resp. $p^{(2)}(u,x_1,x_2)$) be  the probability distribution that maximizes $f^p_\ell(C^{(1)},p(u,x_1,x_2))$ (resp. $f^p_\ell(C^{(2)},p(u,x_1,x_2))$). Let $p_Q(1)=\alpha$ and $p_Q(2)=1-\alpha$, and define $p_{UX_1X_2|Q}(u,x_1,x_2|1)=p^{(1)}(u,x_1,x_2)$ and $p_{UX_1X_2|Q}(u,x_1,x_2|2)=p^{(2)}(u,x_1,x_2)$. Then we have the following chain of inequalities.
\begin{eqnarray*}
&&f_\ell(C^\star)\\
&&<\alpha f_\ell(C^{(1)})+(1-\alpha)f_\ell(C^{(2)})\\
&&\leq\min\! \left\{\!\!\!\begin{array}{l}2C^\star-I(X_1;X_2|UQ)\\C^\star+I(X_1;Y|X_2UQ)\\C^\star+I(X_2;Y|X_1UQ)\\\frac{1}{2}(2C^\star\!+\!I(X_1X_2;Y|UQ)\!-\!I(X_1;X_2|UQ))\\I(UX_1X_2;Y|Q)\end{array}\!\!\!\!\!\right\}\\
&&\leq\min\! \left\{\!\!\!\begin{array}{l}2C^\star-I(X_1;X_2|UQ)\\C^\star+I(X_1;Y|X_2UQ)\\C^\star+I(X_2;Y|X_1UQ)\\\frac{1}{2}(2C^\star\!+\!I(X_1X_2;Y|UQ)\!-\!I(X_1;X_2|UQ))\\I(UX_1X_2;Y)\end{array}\!\!\!\!\!\right\}\\
&&\stackrel{(a)}{\leq} f_\ell(C^\star).
\end{eqnarray*}
Step $(a)$ follows by renaming $(U,Q)$ a $U$ and comparing the left hand with the lower bound characterization of Theorem~\ref{lower bound}. So we have a contradiction, and the proposition is proved.
\end{proof}

\section{The Gaussian MAC}
\label{exGauss}
The output of the Gaussian MAC is given by
$$Y=X_1+X_2+Z$$ where $Z\sim\mathcal{N}(0,1)$, and the transmitters have average power constraints $P_1,P_2$; i.e., $\frac{1}{n}\sum_{i=1}^n\mathbb{E}(X_{1,i}^2)\leq P_1$ and $\frac{1}{n}\sum_{i=1}^n\mathbb{E}(X_{2,i}^2)\leq P_2$. 

We  specialize Theorem \ref{thmupp}  to  obtain an upper bound on the achievable rate $R$. To simplify the bound, pick $U$ to be a noisy version of $Y$ in the form of $U=Y+Z_N$, where $Z_N$ is a Gaussian noise with zero mean and variance $N$ (to be optimized later).  Inequalities \eqref{thmupp1}-\eqref{thmupp5} are upper bounded as follows using maximum entropy lemmas:\allowdisplaybreaks
\begin{eqnarray}
&&\hspace{-1.25cm}R\leq C_1+C_2\label{gauss1}\\
&&\hspace{-1.25cm}R\leq C_2+\frac{1}{2}\log \left(1+P_1(1-\rho^2)\right)\\
&&\hspace{-1.25cm}R\leq C_1+\frac{1}{2}\log \left(1+P_2(1-\rho^2)\right)\\
&&\hspace{-1.25cm}R\leq \frac{1}{2}\log \left(1+P_1+P_2+2\rho\sqrt{P_1P_2}\right)\\
&&\hspace{-1.25cm}2R\stackrel{(a)}{\leq}\!\! \left(\!\!\!\!\begin{array}{l}C_1\!+\!C_2\!+\!\frac{1}{2}\log\!\left(\!{1\!+\!P_1\!+\!P_2\!+\!2\rho\sqrt{P_1P_2}}\right)\\+\frac{1}{2}\log\!\! \left(\!\frac{(1+N+P_1(1-\rho^2))(1+N+P_2(1-\rho^2))}{(1+N+P_1+P_2+2\rho\sqrt{P_1P_2})(1+N)}\!\right)\end{array}\!\!\!\!\right).\label{gauss5}
\end{eqnarray}
To obtain inequality $(a)$ above, write the RHS of \eqref{thmupp5} as 
\begin{eqnarray*}
\begin{array}{l}
C_1+C_2+h(Y|U)-h(YU|X_1X_2)+h(U|X_1)\\+h(U|X_2)-h(U|X_1X_2).
\end{array}
\end{eqnarray*} 
The negative terms are easy to calculate because of the Gaussian nature of the channel and the choice of $U$. The positive terms are bounded from above using the conditional version of the maximum entropy lemma \cite{Thomas87}. It remains to solve a max-min problem (max over $\rho$ and min over $N$). So the rate $R$ is achievable only if there exists some $\rho>0$ for which for every $N>0$ inequalities \eqref{gauss1}-\eqref{gauss5} hold. 

We choose $N$ to be (see \cite[eqn. (21)]{KangLiu11})
\begin{eqnarray}N=\left(\sqrt{P_1P_2}\left(\frac{1}{\rho}-\rho\right)-1\right)^+.\label{choiceofN}
\end{eqnarray} 
\begin{figure}[t!]
\begin{center}
\includegraphics[width=.48\textwidth]{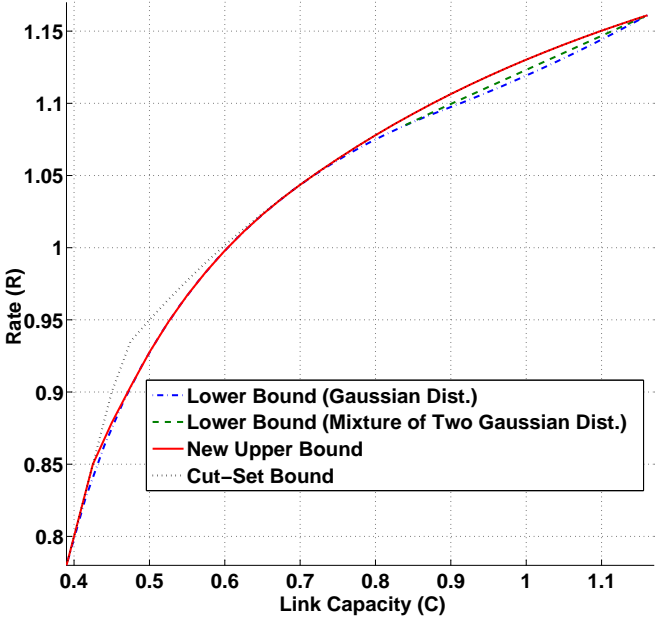}
\end{center}
\caption{Upper and lower bounds on $R$ as functions of $C$ for Gaussian MAC with $P_1=P_2=1$.}
\label{closeup}
\end{figure}

Let us first motivate this choice. It is easy to see that inequality \eqref{gauss5} is simply 
\begin{eqnarray}
I(X_1X_2;Y)-I(X_1;X_2)+I(X_1;X_2|U)
\end{eqnarray} 
evaluated for the joint Gaussian distribution $p(x_1,x_2)$ with covariance matrix 
\begin{eqnarray}
\left[\begin{array}{cc}P_1&\rho\sqrt{P_1P_2}\\\rho\sqrt{P_1P_2}&P_2\end{array}\right].
\end{eqnarray} 
The choice \eqref{choiceofN} makes $U$ satisfy the Markov chain $X_1-U-X_2$ for the regime where 
\begin{eqnarray}
\sqrt{P_1P_2}\left(\frac{1}{\rho}-\rho\right)-1\geq0
\end{eqnarray} and thus minimizes the RHS of \eqref{gauss5}.  Otherwise, $U=Y$ is chosen which results in a redundant bound.
The upper bound is summarized in Corollary \ref{uppGauss}.
\begin{corollary}
\label{uppGauss}
Rate $R$ is achievable only if there exists $\rho\geq0$ such that 
\begin{eqnarray*}
\begin{array}{lll}
&&\ \ \rho\leq \sqrt{1+\frac{1}{4P_1P_2}}-\sqrt{\frac{1}{4P_1P_2}},\\
&&\begin{array}{lll}R&\leq& C_1+C_2\\
R&\leq& C_2+\frac{1}{2}\log\left(1+P_1(1-\rho^2)\right)\\
R&\leq& C_1+\frac{1}{2}\log\left(1+P_2(1-\rho^2)\right)\\
R&\leq& \frac{1}{2}\log \left(1+P_1+P_2+2\rho\sqrt{P_1P_2}\right)\\
2R&\leq& C_1+C_2+\frac{1}{2}\log \left(1\!+\!P_1\!+\!P_2\!+\!2\rho\sqrt{P_1P_2}\right)\\&&-\frac{1}{2}\log\left(\frac{1}{1-\rho^2}\right),
\end{array}\\
&&\hspace{-.65cm}\text{or}\\
&&\ \  \sqrt{1+\frac{1}{4P_1P_2}}-\sqrt{\frac{1}{4P_1P_2}}\leq\rho\leq 1,\\
&&\begin{array}{lll}
R&\leq& C_1+C_2\\
R&\leq& C_2+\frac{1}{2}\log\left(1+P_1(1-\rho^2)\right)\\
R&\leq& C_1+\frac{1}{2}\log\left(1+P_2(1-\rho^2)\right)\\
R&\leq& \frac{1}{2}\log \left(1+P_1+P_2+2\rho\sqrt{P_1P_2}\right).
\end{array}
\end{array}
\end{eqnarray*}
\end{corollary}


For a lower bound, we use Theorem \ref{lower bound} and choose $(U,X_1,X_2)$ to be jointly Gaussian with mean $0$ and covariance matrix $\mathbf{K}_{UX_1X_2}$. 
Fig. \ref{closeup} shows the upper and lower bound as a function of $C$ for a symmetric network with $C_1=C_2=C$ and $P_1=P_2=1$. The dash-dotted curve in Fig. \ref{closeup} shows the rates achievable  using the scheme of Section \ref{secach} with a joint Gaussian distribution.  It is interesting to see that the obtained lower bound is not concave in $C$. This does not contradict Proposition \ref{thmconcave} because  Gaussian distributions are sub-optimal. 
The improved dashed curve shows rates that are achievable using a mixture of two Gaussian distributions.

For symmetric diamond networks where $C_1=C_2=C$ and $P_1=P_2=P$, we specify a regime of $C$ for which the lower and upper bounds meet and thus characterize the capacity. This is summarized in Theorem \ref{gaussmatch}. See \cite{reportMAC} for the proof.
\begin{theorem}
\label{gaussmatch}
For a symmetric Gaussian diamond network, the upper bound in Theorem \ref{thmupp} meets the lower bound if $C\leq \frac{1}{4}\log(1+2P)$, $C\geq \frac{1}{2}\log(1+4P)$, or
\begin{eqnarray}
\label{thmC}
\frac{1}{4}\log\frac{1+2P(1+\rho^{(1)})}{1-{\rho^{(1)}}^2}\leq C\leq \frac{1}{4}\log\frac{1+2P(1+\rho^{(2)})}{1-{\rho^{(2)}}^2}
\end{eqnarray}
where 
\begin{eqnarray}
&&\rho^{(1)}=\frac{-(1+2P)+\sqrt{12P^2+(1+2P)^2}}{6P}\label{r1}\\
&&\rho^{(2)}=\sqrt{1+\frac{1}{4P^2}}-\frac{1}{2P}.\label{r2}
\end{eqnarray}
\end{theorem}
 For example, the upper and lower bounds match in Fig. \ref{closeup} (where $P=1$) for $C\leq 0.3962$, $0.4807\leq C \leq 0.6942$, and $C\geq1.1610$. 

\begin{remark}
There is a close connection between the bound of Corollary \ref{uppGauss} and that of \cite{KangLiu11}. 
The upper bound in \cite{KangLiu11} is tighter than Corollary \ref{uppGauss} in certain regimes of operation. For example in Fig. \ref{closeup}, the upper bound of \cite{KangLiu11} matches the lower bound for all $C\leq 0.6942$. The reason seems to be that the methods used in \cite{KangLiu11} provide a better single letterization of the upper bound for the Gaussian MAC. 
\end{remark}



\section{The binary Adder Channel}
\label{exAdd}
Consider the binary adder channel  defined by $\mathcal{X}_1=\{0,1\}$, $\mathcal{X}_2=\{0,1\}$, $\mathcal{Y}=\{0,1,2\}$, and $Y=X_1+X_2$. The best known upper bound for this channel is the cut-set bound. We  specialize Theorem \ref{thmupp} to derive a new upper bound on the achievable rate.  For simplicity, we assume $C_1=C_2=C$.

Define $p(u|x_1,x_2,y)$ to be a symmetric channel as shown in Fig. \ref{channelu}, with parameter $\alpha$ to be optimized. From Theorem~\ref{thmupp}, the rate $R$ is  achievable only if for some pmf $p(x_1,x_2)$ inequalities \eqref{thmupp1}-\eqref{thmupp5} hold for the aforementioned choice of $U$ and for every $\alpha$; i.e., we have to solve a max-min problem (max over $p(x_1,x_2)$, min over $\alpha$).
\begin{figure}[t!]
\centering
\begin{tikzpicture}[scale=.65]
\tikzstyle{every node}=[draw,shape=circle,font=\small];

\path (-.5,4.5) node[draw=none] () {$Y$};
\path (3.5,4.5) node[draw=none] () {$U$};

\coordinate (y0) at (0,0);
 \fill (y0) circle (4pt);
 \coordinate (y1) at (0,1.875);
 \fill (y1) circle (4pt);
 \coordinate (y2) at (0,3.75);
 \fill (y2) circle (4pt);

\path (-.5,0) node[draw=none]() {$2$};
\path (-.5,1.875) node[draw=none] () {$1$};
\path (-.5,3.75) node[draw=none] () {$0$};

 \coordinate (u0) at (3,1);
 \fill (u0) circle (4pt);
 \coordinate (u1) at (3,2.5);
 \fill (u1) circle (4pt);
 \path (3.5,1)  node[draw=none] () {$1$};
\path (3.5,2.5) node[draw=none] () {$0$};

\draw[->] (y0) --node[left,draw=none,xshift=.4cm,yshift=-.4cm]{$1-\alpha$} (u0);
\draw[->] (y0) --node[left,draw=none,xshift=0cm,yshift=-.4cm]{$\alpha$} (u1);
\draw[->] (y1) --node[left,draw=none,xshift=0cm,yshift=-0.15cm]{$\frac{1}{2}$} (u0);
\draw[->] (y1) --node[left,draw=none,xshift=0cm,yshift=0.15cm]{$\frac{1}{2}$} (u1);
\draw[->] (y2) --node[left,draw=none,xshift=.4cm,yshift=.4cm]{$1-\alpha$} (u1);
\draw[->] (y2) --node[left,draw=none,xshift=0cm,yshift=.45cm]{$\alpha$} (u0);

\end{tikzpicture}
\caption{Channel $p(u|x_1,x_2,y)$}
\label{channelu}
\end{figure}
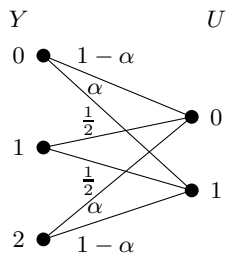

In the rest of this Section we manipulate the above bound to formulate a simplified upper bound. First, loosen the upper bound by looking at the min-max problem (rather than the max-min problem). Fix $\alpha$. For a fixed channel $p(u|x_1,x_2,y)$, the upper bound is concave in $p(x_1,x_2)$ (see Remark \ref{concave}). The concavity, together with the symmetry of the problem and the auxiliary channel imply the following lemma. 
\begin{lemma}
\label{lemmadoubly}
For a fixed $\alpha$, the optimizing pmf is that of a doubly symmetric binary source.
\end{lemma}
Using Lemma \ref{lemmadoubly}, $p(x_1,x_2)$ may be assumed to be a doubly symmetric binary source with cross-over probability $p$. The bounds in \eqref{thmupp1}-\eqref{thmupp5} reduce to
\begin{eqnarray}
\begin{array}{lll}
R&\leq&2C\\
R&\leq&C\!+\!h_2(p)\\
R&\leq&h_2(p)\!+\!1\!-\!p\\
2R&\stackrel{(a)}{\leq}&2C\!+\!2h_2(p)\!-\!p+I(X_1;X_2|U).
\end{array}
\end{eqnarray}
The inequality in $(a)$ follows from \eqref{equivthm5} and the fact that $(X_1X_2)-Y-U$ forms a Markov chain.

To obtain the best bound, we choose $U$ such that $X_1-U-X_2$ forms a Markov chain; i.e., the optimal $\alpha$ satisfies
\begin{eqnarray}
\label{choicea}
\alpha(1-\alpha)=\left(\frac{p^\star}{2(1-p^\star)}\right)^2
\end{eqnarray}
where $p^\star$ is the $p$ which maximizes
\begin{eqnarray}
\min\left\{\!\!\!
\begin{array}{l}
2C,
C+h_2(p),
h_2(p)+1-p,
2C+2h_2(p)-p
\end{array}\!\!\!\right\}.
\end{eqnarray}
Since we have $p^\star\leq \frac{1}{2}$, the choice \eqref{choicea} is valid and we therefore obtain Corollary \ref{corbinaryadder}.
\begin{corollary}
\label{corbinaryadder}
Rate $R$ is achievable only if there exists some $p$, $0\leq p\leq\frac{1}{2}$, such that
\begin{eqnarray}
\begin{array}{lll}
R&\leq&2C\\
R&\leq&C\!+\!h_2(p)\\
R&\leq&h_2(p)\!+\!1\!-\!p\\
2R&\leq&2C\!+\!2h_2(p)\!-\!p.
\end{array}
\end{eqnarray}
\end{corollary}

Fig. \ref{plotbinaryadder} plots this bound (Corollary \ref{corbinaryadder}) and compares it with the cut-set bound and the lower bound of Theorem \ref{lower bound} for different values of $C$. It turns out that the lower and upper bounds meet for $C\leq .75$ and $C\geq .7929$. See \cite{reportMAC} for the proof.
\begin{figure}[t!]
\begin{center}
\includegraphics[width=.42\textwidth]{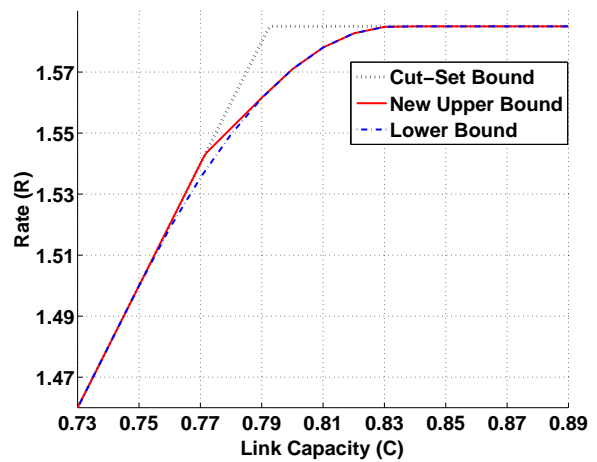}
\end{center}
\caption{Lower and upper bounds on $R$ as functions of $C$ for the symmetric binary adder MAC.}
\label{plotbinaryadder}
\end{figure}
\section{Acknowledgment}
The work of S. Saeedi Bidokhti was supported by the Swiss National Science Foundation fellowship no. 146617. The work of G. Kramer was supported by an Alexander von Humboldt Professorship endowed by the German Ministry of Education and Research.
\bibliographystyle{IEEEtran}
\bibliography{bibliographyMAC}

\begin{thebibliography}{1}
\providecommand{\url}[1]{#1}
\csname url@samestyle\endcsname
\providecommand{\newblock}{\relax}
\providecommand{\bibinfo}[2]{#2}
\providecommand{\BIBentrySTDinterwordspacing}{\spaceskip=0pt\relax}
\providecommand{\BIBentryALTinterwordstretchfactor}{4}
\providecommand{\BIBentryALTinterwordspacing}{\spaceskip=\fontdimen2\font plus
\BIBentryALTinterwordstretchfactor\fontdimen3\font minus
  \fontdimen4\font\relax}
\providecommand{\BIBforeignlanguage}[2]{{%
\expandafter\ifx\csname l@#1\endcsname\relax
\typeout{** WARNING: IEEEtran.bst: No hyphenation pattern has been}%
\typeout{** loaded for the language `#1'. Using the pattern for}%
\typeout{** the default language instead.}%
\else
\language=\csname l@#1\endcsname
\fi
#2}}
\providecommand{\BIBdecl}{\relax}
\BIBdecl

\bibitem{Schein01}
B.~E. Schein, \emph{Distributed coordination in network information
  theory}.\hskip 1em plus 0.5em minus 0.4em\relax PhD Dissertation, MIT, 2001.

\bibitem{AvestimehrDiggaviTse11}
A.~Avestimehr, S.~Diggavi, and D.~Tse, ``Wireless network information flow: A
  deterministic approach,'' \emph{IEEE Trans. Inf. Theory}, vol.~57, no.~4, pp.
  1872 --1905, april 2011.

\bibitem{TraskovKramer07}
D.~Traskov and G.~Kramer, ``Reliable communication in networks with
  multi-access interference,'' in \emph{Proc. Inf. Theory Workshop}, 2007.

\bibitem{KangLiu11}
W.~Kang and N.~Liu, ``The gaussian multiple access diamond channel,'' in
  \emph{Proc. Int. Symp. Inf. Theory}, 2011.

\bibitem{Ozarow80}
L.~Ozarow, ``On a source-coding problem with two channels and three
  receivers,'' \emph{Bell System Technical Journal}, vol.~59, no.~10, pp.
  1909--1921, 1980.

\bibitem{reportMAC}
\BIBentryALTinterwordspacing
S.~{Saeedi Bidokhti} and G.~Kramer, ``Capacity bounds for a class of diamond
  networks,'' Jan. 2014. [Online]. Available:
  \url{https://mediatum.ub.tum.de/doc/1189765/1189765.pdf}
\BIBentrySTDinterwordspacing

\bibitem{ElGamalKim}
A.~{El Gamal} and Y.~H. Kim, \emph{Network Information Theory}.\hskip 1em plus
  0.5em minus 0.4em\relax Cambridge University Press, 2011.

\bibitem{Thomas87}
J.~Thomas, ``Feedback can at most double gaussian multiple access channel
  capacity (corresp.),'' \emph{IEEE Trans. Inf. Theory}, vol.~33, no.~5, pp.
  711--716, 1987.

\end{thebibliography}

%
%
%

\end{document}